\documentclass[11pt,reqno]{article}

\usepackage{hyperref}
\usepackage{phaistos}
\usepackage{manfnt}
\usepackage{graphicx}   % need for figures
\usepackage{verbatim}   % useful for program listings
\usepackage{amsmath,amssymb,amsthm}  
\newtheorem{theorem}{Theorem}[section]

\newtheorem{conjecture}[theorem]{Conjecture}
\newtheorem{Theorem}{Theorem}
\newtheorem{Lemma}[theorem]{Lemma}

\newtheorem{Corollary}[theorem]{Corollary}

\theoremstyle{definition}

\newtheorem{Definition}[theorem]{Definition}
\newtheorem{definition}[theorem]{Definition}

\usepackage{mathtools}
\usepackage{verbatim}
\usepackage{amsfonts}
\usepackage{protosem}
\usepackage{mathabx}

\newcommand{\R}{{\mathbb R}}
\newcommand{\C}{{\mathbb C}}

\DeclareMathOperator*{\Ex}{\mathbf E}

\newcommand{\wtilde}{\widetilde}

\usepackage{titlesec}

\usepackage[left=1.25in, right=1.25in, top=1.35in, bottom=0.8in, headsep=1.2cm]{geometry}
\titlespacing*{\section}{0pt}{7pt}{4pt}
\titlespacing*{\subsection}{0pt}{3pt}{2pt}
\titlespacing*{\paragraph}{0pt}{5pt}{3pt}

\titleformat{\section}[block]
{\LARGE\bfseries}
{\thesection}{10pt}{}{}

\titleformat{\subsection}[block]
{\Large \sffamily}
{\thesubsection}{10pt}{}[]

\titleformat{\paragraph}[runin]
{\large \bfseries}
{}{}{}[]

\begin{document}
\begin{titlepage}

\title{\bfseries \LARGE Tighter Relations Between Sensitivity and Other Complexity Measures}

\author{
Andris Ambainis\thanks{Supported by the European Commission under the project QALGO
(Grant No. 600700) and the ERC Advanced Grant MQC (Grant No. 320731).}\\ {\small U. of Latvia}\ \and
 Mohammad Bavarian\thanks{Supported by NSF through STC Award 0939370, and CCF-1065125.} \\ {\small MIT}  \and
 Yihan Gao \\ {\small UIUC} \and
 Jieming Mao  \\ {\small Princeton} \and
 Xiaoming Sun\thanks{Supported in part by the National Natural Science Foundation of China Grant 61170062, 61222202.} \\ {\small Chinese Academy of Sciences}
 \and  Song Zuo \\ {\small Tsinghua}}
\date{}

\maketitle

\begin{abstract}
Sensitivity conjecture is a longstanding and fundamental open problem in the area of complexity measures of Boolean functions and decision tree complexity. The conjecture postulates that the maximum sensitivity of a Boolean function is polynomially related to other major complexity measures. Despite much attention to the problem and major advances in analysis of Boolean functions in the past decade, the problem remains wide open with no positive result toward the conjecture since the work of Kenyon and Kutin from 2004~\cite{KK}.

In this work, we present new upper bounds for various complexity measures in terms of sensitivity improving the bounds provided by Kenyon and Kutin.
Specifically,  we show that $\deg(f)^{1-o(1)}=O(2^{s(f)})$ and $C(f) \leq 2^{s(f)-1} s(f)$; these in turn imply various corollaries regarding the relation between sensitivity and other complexity measures, such as block sensitivity, via known results.
The gap between sensitivity and other complexity measures remains exponential but these results are the first
improvement for this difficult problem that has been achieved in a decade.

 \end{abstract}
 
\end{titlepage}

\section{Introduction}
In this paper, we are concerned with a fundamental and challenging open problem in complexity theory known as the \emph{Sensitivity Conjecture} (also called sensitivity vs.~block sensitivity problem). Since its appearance in Nisan and Szegedy~\cite{Nisan94}, the problem has received a considerable amount of attention from numerous researchers (see \cite{Aaronson10, HKP11}). By now many equivalent formulations and connections between this conjecture and other unsolved problems in combinatorics and analysis of Boolean functions have been discovered which has resulted in an increase in the prominence and popularity of the conjecture.

The conjecture originates from the theory of complexity measures of Boolean functions and decision tree complexity. The basic object of study in this area is \emph{decision tree complexity} of Boolean functions and also its \emph{randomized} or \emph{quantum} variants.\footnote{Recall that given a Boolean function $f:\{0,1\}^n\rightarrow \{0,1\}$, decision tree complexity of $f$ refers to the minimum number of queries that an algorithm querying the input variables $(x_1,x_2, \ldots, x_n)$ must make to be able to successfully compute $f$ on every input. The reason for the name \emph{decision tree} is that any query algorithm can be identified with a directed tree with inner vertices labeled by the input variables, directed edges corresponding to the value of the variable just read while each leaf contains the value outputted by the algorithm upon reaching that leaf. With this picture, the query complexity of the algorithm exactly corresponds to the \emph{depth} of the tree.} The study of decision tree complexity is directly connected (and is usually done via) the study of more combinatorial and analytic measures of complexity of Boolean functions such as Fourier degree, block sensitivity, certificate complexity and etc. 
Since the time of Nisan and Szegedy~\cite{Nisan94}, it was known that the above and all other major complexity measures of Boolean functions are polynomially related to one another. The only major exception to the above principle is the \emph{maximum sensitivity} which is still unknown to be polynomially related to other complexity measures. The sensitivity conjecture is precisely the statement that the above principle also holds in the case of maximum sensitivity.

\begin{conjecture}[sensitivity conjecture]
There exists a constant  $d\in \R^+$ such that for any Boolean function $f:\{-1,1\}^n \rightarrow \{-1,1\}$ we have 
\[ bs(f)=O(s(f)^d),\]
where $s(f)$ and  $bs(f)$ denote the sensitivity and the block sensitivity (defined in Section~\ref{sec:preliminary}) of the function $f$ .
\end{conjecture}

Let us note that in the formulation of the conjecture, we could have opted to replace the block sensitivity with any other widely used complexity measure of Boolean functions (such as Fourier degree, deterministic decision tree complexity, etc.) because as we mentioned before, all these are polynomially related to block sensitivity~\cite{BW02, Nisan94}.

\subsection{Prior work}
As discussed above, through the work of various researchers by now many different equivalent forms of the sensitivity conjecture are available. Almost all of these different formulations and various approaches to the conjecture are discussed in the excellent survey of P.~Hatami et al.~\cite{HKP11} (see also the blogpost of Aaronson~\cite{Aaronson10}). We refer to these works for a more detailed exposition of the background. We briefly recall some of the more immediately relevant facts:

The best known upper bound on block sensitivity is
\begin{equation} \label{eqn:best_previous_bound}
bs(f) \leq (\frac{e}{\sqrt{2\pi}}) e^{s(f)} \sqrt{s(f)},
\end{equation}
given by Kenyon and Kutin \cite{KK}. In the other direction, the first progress on the lower bound was made by Rubinstein~\cite{Rubinstein95} who gave the first quadratic separation for block sensitivity and sensitivity by constructing a Boolean function $f$ with $bs(f)={1\over 2}s(f)^2$. Currently, the best lower bound is due to Ambainis and Sun who in~\cite{AS11} exhibited a function with $bs(f) = {2\over 3}s(f)^2-{1\over 2}s(f)$.

\subsection{Our results} Our first result in this paper is the following estimate regarding the relation between the maximum sensitivity and Fourier degree of a Boolean function.
\begin{Theorem}\label{sens_deg1} Let $f:\{0,1\}^n\rightarrow \{-1,1\}$ be a Boolean function. Then
\[ \deg(f)^{1-o(1)}\leq  \, s(f)2^{s(f)} \, ,\]
where $o(1)$ denotes a term that vanishes as $\deg(f)\rightarrow \infty$.
\end{Theorem}
The proof of the above theorem is a mixture of techniques from Fourier analysis and combinatorics.
The argument is partly inspired by the arguments in the paper of Chung et al.~\cite{Chung88} which recently played an important role in~\cite{Aaronson13} where the query complexity of partial functions coming from the restrictions of parity function was studied.

For sensitivity versus certificate complexity, we can prove a somewhat stronger theorem which has direct consequences for sensitivity versus block sensitivity problem (which is the original formulation of sensitivity conjecture by Nisan and Szegedy \cite{Nisan94}).
\begin{Theorem}\label{thm:C(f)}
 For any Boolean function $f$,
\begin{equation}
C_1(f)\leq 2^{s_0(f)-1} s_1(f),\ \ \
C_0(f)\leq 2^{s_1(f)-1} s_0(f).
\end{equation}
\end{Theorem}
Here $C_0(f)$ and $C_1(f)$ denote the $0$-certificate complexity and $1$-certificate complexity of $f$. These notions -- along with the rest of the background material on complexity measures of Boolean functions -- are presented in Section \ref{sec:preliminary}.

Using the known relations between various complexity measures of Boolean functions, we can derive several consequences from the above result. 
\begin{Corollary}
\label{cor:1}
For any Boolean function $f$,
\begin{equation*}
bs(f)\leq C(f) \leq 2^{s(f)-1} s(f).
\end{equation*}
\end{Corollary}
Combining Theorem~\ref{thm:C(f)} and some previous results, we can also give another upper bound for block sensitivity.
\begin{Corollary}\label{thm:bs(f)}
For any Boolean function $f$,
\begin{equation}
bs(f)\leq \min\{2^{s_0(f)},2^{s_1(f)}\} s_1(f)s_0(f).
\end{equation}
\end{Corollary}
Hence, we see that our Theorems \ref{sens_deg1} and \ref{thm:C(f)} and their corollaries show an improved exponent in relation between sensitivity and various complexity measures of Boolean functions compared to the previous best bound shown in equation (\ref{eqn:best_previous_bound}). Beside being the first positive result toward the sensitivity conjecture since the work of Kenyon and Kutin from 2004,  we believe our results have the merit of introducing new ideas and techniques which could be valuable elesewhere as well as in the future works  on this fundamental conjecture.

Although the bounds obtained in Theorems \ref{sens_deg1} and Theorem \ref{thm:C(f)} look quite similar, the theorems do not follow one from one another by using the known
relations between certificate complexity and Fourier degree. On the contrary, the two theorems
are obtained by using rather different techniques. However, we shall note that despite their differences both proofs of Theorem \ref{sens_deg1} and \ref{thm:C(f)} crucially rely on the small set expansion properties of Boolean hypercube. It is plausible that better analytic estimates along the lines of \cite{Samor2007} could be useful to slightly improve our results--- though a significant improvement is likely to require new ideas.

\paragraph{Organization.} In Section~\ref{sec:preliminary} we recall some basic definitions and concepts relevant to this work. In Section \ref{sens_degree_section}, we prove Theorem~\ref{sens_deg1} and in Section~\ref{sec:thm1}, we prove Theorem \ref{thm:C(f)} and its corollaries. Both Sections~\ref{sens_degree_section} and~\ref{sec:thm1} are self-contained and can be read in any order.

\section{Preliminaries}\label{sec:preliminary}

In this paper, we work with total Boolean functions over the hypercube and their measures of complexity. or completeness, we briefly recall some basic definitions. For more information regarding the complexity measures and their relations we recommend the survey \cite{BW02}.

We work with the usual graph structure on the hypercube by connecting $x,y\in \{0,1\}^n$ if and only if $x,y$ differ in a single coordinate.
We always denote by $\log(\cdot)$ the logarithm with the base $2$.

\begin{Definition}\label{def:sensitivity}
The pointwise sensitivity $s(f,x)$ of a function $f$ on input $x$ is defined as the number of bits on which the function is sensitive, i.e.
\[ s(f,x)=\big|\{i\in[n]|f(x)\neq f(x^{(i)})\}\big|, \]
where $x^{(i)}$ is obtained by flipping the $i$-th bit of $x$. We define the total sensitivity by
\[s(f)=\max\big\{s(f,x)|x\in\{0,1\}^n\big\}\, ,\]
and the  0-sensitivity and 1-sensitivity by
\[ s_0(f)=\max\big\{s(f,x)|x\in\{0,1\}^n,f(x)=0\big\}, \quad s_1(f)=\max\big\{s(f,x)|x\in\{0,1\}^n,f(x)=1\big\}\, .\]
\end{Definition}
Block sensitivity is another important complexity measure which is obtained by relaxing the requirement that we have to only flip single coordinates by allowing flipping disjoint blocks. More formally block sensitivity is defined as follows:

\begin{Definition}
The pointwise block sensitivity $bs(f,x)$ of $f$ on input $x$ is defined as maximum number of pairwise disjoint subsets $B_1,...,B_k$ of $[n]$ such that $f(x) \neq f(x^{(B_i)})$ for all $i\in [k]$. Here $x^{(B_i)}$ is obtained by flipping all the bits $\{x_j|j\in B_i\}$ of $x$. Define the block sensitivity of $f$ as \[bs(f)=\max\big\{bs(f,x)|x\in\{0,1\}^n\big\},\] and
the 0-block sensitivity and 1-block sensitivity, analogously to Definition  \ref{def:sensitivity}, by
\[bs_0(f)=\max\big\{bs(f,x)|x\in\{0,1\}^n,f(x)=0\big\}, \quad bs_1(f)=\max\big\{bs(f,x)|x\in\{0,1\}^n,f(x)=1\big\}\, .\]
\end{Definition}
The certificate complexity is another very useful complexity measure with a more non-deterministic type of definition. It is defined as follows:

\begin{Definition}\label{cer_comp:def}
The  certificate complexity $C(f,x)$ of $f$ on input $x$ is defined as the minimum length of a partial assignment of $x$ such that $f$ is constant on this restriction. Define the certificate complexity of $f$ by \[C(f)=\max\big\{C(f,x)|x\in\{0,1\}^n\big\},\]
and
the 0-certificate and 1-certificate by  \[C_0(f)=\max\big\{C(f,x)|x\in\{0,1\}^n,f(x)=0\big\}, \quad C_1(f)=\max\big\{C(f,x)|x\in\{0,1\}^n,f(x)=1\big\}.\]
\end{Definition}

Another important notion for us is Fourier degree. It is also polynomially related to block-sensitivity and certificate complexity. To define Fourier degree, recall that any function $f:\{0,1\}^n\rightarrow \C$ can be expanded in terms of Fourier characters as follows
\[ f(x)=\sum_{S\subseteq [n]} \hat{f}(S) \, \chi_S(x)\,  , \]
where $\chi_S(x)=(-1)^{\sum_{i\in S} x_i}$.
\begin{definition} Let $f:\{-1,1\}^n \rightarrow \R$ and let $\hat{f}(\cdot)$ denote its Fourier transform. We define Fourier degree of $f$ by
\[ \deg(f)= \max_{\hat{f}(S)\neq 0} \,|S| \, .\]
\end{definition}
Finally, we mention an important and well-known combinatorial result over the hypercube, usually attributed to Harper \cite{Harper66}. 
\begin{Lemma}[Hamming cube isoperimetry \cite{Harper66}]\label{iso_theorem}
Assume $\emptyset \neq A\subseteq\{0,1\}^n$. Let $d$ be the average degree of vertices of $A$ with graph structure on $A$ induced from the natural Hamming graph of $\{0,1\}^n$. Then we have
\[ |A| \geq 2^{d}\, .\]
\end{Lemma}
The above lemma is quite easy to prove by induction. For a detailed proof which covers the application to combinatorics, we recommend consulting the book by Bollob\'{a}s \cite{Bol}. 

The above theorem implies that if $|A|$ is small, the average degree $d$ must also be relatively small. In this case, the ratio between the number of the edges leaving the set $A$ and the total number of incident edges to $A$, which is equal to $1-d/n$, is relatively large. This justifies the alternative name given to the above theorem as the ``small set expansion" property of the Hamming cube.

In Section \ref{sec:thm1}, we need an equivalent formulation of discrete isoperimetric inequality, Lemma \ref{iso_theorem},  which will be a more convenient for our purposes there. 

\begin{Lemma}\label{lem:iso}
For any $A\subseteq \{0,1\}^n$, the edges between $A$ and $\bar{A}=\{0,1\}^n\setminus A$ is lower bounded by
$$|E(A,\bar{A})|\geq |A|(n-\log_2 |A|).$$
\end{Lemma}

\section{Sensitivity versus degree}\label{sens_degree_section}
In this section, we shall prove Theorem \ref{sens_deg1}. Let $f:\{0,1\}^n\rightarrow \{-1,1\}$ be a Boolean function. To prove Theorem \ref{sens_deg1}, the key idea is to count the following objects.

\begin{Definition} An $(l, r)$ S-triple consists of a point $x\in \{0,1\}^n$ and two sets $L\subseteq R\subseteq [n]$ with $|L|= l$ and $|R|=r$ such that $f(x)\neq f(x^i)$ for any $i\in L$.
\end{Definition}
In our application, the two parameters $l \leq r$ are chosen as follows: $l=c \log r$, for some $c>0$  an appropriately chosen constant, and $r$ will be a slowly growing function of $n$ which will be asymptotically $\log\log n$. The upper bound on the number of S-triples is easy to establish.
\begin{Lemma}\label{upper_lemma}
The number of $(l,r)$ S-triples is less than or equal to
\[2^{n} \frac{ {s(f)}^{l} \, n^{r-l}} {l ! (r-l)! }. \]
\end{Lemma}
\begin{proof} We can assume $s(f)\geq l$ as otherwise the number of S-triples is zero. Consider any $x\in \{0,1\}^n$. The number of S-triples involving $x$ is bound by $\max_{1\leq q\leq s(f)} \binom{q}{l} \binom{n-l}{ r-l}$. This is clearly bounded by $\frac{ {s(f)}^{l} \, n^{r-l}} {l ! (r-l)! }$ which implies the above lemma.
\end{proof}
Now we are in a position to layout the structure of the proof.

\subsection{The overall structure of the proof}
The technical of proving Theorem \ref{sens_deg1} is to prove a lower bound on the number of S-triples which coupled with the above lemma gives the desired lower bound on $s(f)$. To do so, we shall need the following two facts:

\begin{enumerate}
\item A \textbf{weak bound} for $s(f)$ versus $\deg(f)$. For example it follows from the work of Kenyon and Kutin that $\deg(f)\leq 10^{s(f)+1}$. 
\item \textbf{Hypercube isoperimetric inequality} as in Lemma \ref{lem:iso}.
\end{enumerate}
Briefly, the plan is to use the isoperimetric inequality to \emph{boost} the weak bound to our desired bound of $\deg(f)^{1-o(1)} \leq 2^{s(f)}$. The key steps of the arguments are as follows.
\begin{enumerate}
\item[A.] We consider the restriction of the functions $f$ to the subcubes of dimension $r$. For any such restriction, we show the existence of a $(l,r)$ S-tripes consistent with that restriction by applying the weak bound. The precise dependence of $l$ on $r$ is simply dictated by the known weak bound we use. 
\item[B.] We use isoperimetric inequality to show that the existence of a single S-triple consistent with a particular restriction $z$ immediately gives rise to \emph{many} S-triples consistent with the same restriction $z$.
\item[C.] This provides for us a lower bound on the number of S-triples which combined with Lemma \ref{upper_lemma} gives us the desired result.
\end{enumerate}

In order to carry out the argument, we will need a few definition regarding restrictions of functions over the discrete cube which we now present.

\subsection{Restrictions of Boolean functions}
\begin{Definition} Given $z\in \{0,1,*\}^n$ and $R\subseteq [n]$, we call them a \emph{compatible pair} if $R=\{i\in[n] \, :z(i)=*\}$. Each $z\in \{0,1,*\}^n$  naturally corresponds to $|R|$-dimensional subcube $Q_z\subseteq \{0,1\}^n$ defined as follows: 
\[ Q_z= \{ y\in \{0,1\}^n \, : \, z_i \neq * \Rightarrow y_i=z_i \}, \]
i.e. $Q_z$ is constructed by freezing the coordinates of $y$ in $[n]\setminus R$ according to $z$, and letting the rest of coordinates $y_i$ for $i\in R$ to be free.
\end{Definition}
Let  $f:\{0,1\}^n\rightarrow \R$. Given a compatible pair $z\in\{0,1,*\}^n$ and $R=\{ i\in [n] \, : \, z(i)=* \}$ we obtain a restriction function $f|_z$ given by restricting $f$ to $Q_z$
\begin{Definition} Given $z\in \{0,1,*\}^n$ and $x\in \{0,1\}^n$ (here $x$ is not necessarily in $Q_z$), define $y=(x \downarrow z)$ to be projection of $x$ to $Q_z$ given by $y_i= z_i$ for any $i\in [n]$ such that $z(i)\neq *$ and $y_i=x_i$ for all the other $i\in [n]$. We define
\[ f|_z (x) = f(x\downarrow z)\, . \]
\end{Definition}
Notice that $f|_z(x)$ is a function over whole $\{0,1\}^n$ though its value only depends on $R$ the coordinates which $z(i)=*$. Given the above definition one can easily infer the Fourier expansion of the restriction function $f|_z(\cdot)$ from that of $f$ as follows.
\[ \left(f|_z\right)(x)= \, \sum_{S\subseteq R} \chi_S(x) \sum_{U\cap R= S} \hat{f}(U) \chi_{U\backslash S}(z)\, . \]
We need the following lemma regarding the degree of restrictions of a function.
\begin{Lemma} \label{restriction_lemma} Let $f:\{-1,1\}^n\rightarrow\{-1,1\}$ be function of degree $n$. Let $R\subseteq [n]$. Then there exist $z\in \{-1,1,*\}^n$ compatible with $R$ such that $\left(f|_z\right)$ is also full-degree $|R|$.
\end{Lemma}
\begin{proof}
The coefficient of the highest monomial in Fourier expansion of $\left(f|_z\right)$ is given by
\[ \sum_{R\subseteq U} \hat{f}(U) \chi_{U\backslash R} (z). \]
Now we calculate the expectation of the square of this value for a random $z$ compatible with $R$.
\[ \Ex_{z} \left(\sum_{R\subseteq U} \hat{f}(U) \chi_{U\backslash R} (z) \right)^2= \sum_{R\subseteq U}\hat{f}(U) ^2\geq \hat{f}([n])^2 >0  \]
where for the last inequality we used the fact that $f$ is full-degree.
 \end{proof}

The importance of the above lemma is that it allows us to apply the weak bound in line with our \emph{boosting} strategy: Fix some $R\subseteq[n]$. By the lemma above, there exists $z\in\{0,1,*\}^n$ compatible with $R$ such that $f|_{z}$ is full-degree. The importance of existence of $z$ is that $Q_{z}$ always contain an S-triple which was the object we were interested to count. More precisely, since $f|_{z}$ is full-degree by induction on the degree in Theorem \ref{sens_deg1} there exists subset $L\subseteq R$ with $|L| \geq \frac{1}{4}\log |R|$ such that there exist $x\in Q_{z}$ such that $f|_{z}(x)\neq f|_{z}(x^i)$ for every $i\in L$. Taking $l=|L|$ and $r=|R|$ we see that $(x, L, R)$ constitutes an S-triple. We use the existence of $z$ and Harper's lemma \ref{iso_theorem} to prove that for every $R$ there exists not only one such $z$ but in fact many. This is the key estimate we need to prove our result.
\subsection{The main proof of sensitivity versus Fourier degree estimate}

\begin{proof}[Proof of Theorem \ref{sens_deg1}]
Without loss of generality we can assume $f$ is full-degree. If this is not the case, choose $S\subseteq [n]$ with $|S|=\deg(f)$  such that $\hat{f}(S)\neq 0$, then set the coordinates outside $S$ arbitrarily to get a Boolean function on the $|S|$-dimensional hypercube with full-degree. It is enough to prove the statement for this restriction of the original $f$ as  restricting a function can only decrease the sensitivity.

Let $r=\omega(1)$ be a very slowly growing function of $n$ to be specified later. Fix a set $R\subseteq [n]$ with $|R|=r$. By Lemma \ref{restriction_lemma} there exist $z\in \{0,1,*\}^n$ compatible with $R$ such that the restricted function $\left(f|_{z}\right)$ has degree $r$. Now by induction $s(f|_{z})\geq l$ where $l=\Theta(\log r)$. (we can take $l=\frac{1}{3}\log r$ for concreteness.)  This means we can find a point $x\in Q_{z}$ with $l$ neighbors $x_1, x_2, \ldots, x_{l}$ such that
\[ (f|_{z})(x_1)= (f|_{z})(x_2)=\ldots = (f|_{z})(x_l)\neq (f|_{z})(x) \, .\]
Let $L=\{i_1, i_2, \ldots, i_l\} \subseteq R$ be the direction of the edges $(x,x_1), (x,x_2), \ldots, (x,x_l)$ respectively. Then $(x, L,R)$ constitutes a $(l,r)$ S-triple.

So far for any $R\subseteq[n]$ we have shown the existence of one such S-triple. Now we show there are many such triples. Consider $Z_R$ which is the set of all $z\in \{0,1,*\}^n$ compatible with $R$. Notice that $Z_R$ can be naturally associated with a $(n-r)$-hypercube with $z_1,z_2\in Z_R$ said to be \textit{neighbors} in direction $j\in [n]\setminus R$ if  $z_1(i)=z_2(i)$ for $i\in [n]\setminus \{j\}$ and $z_1(j)\neq z_2(j)$. (Clearly $z_1(j)\neq *$ and $z_2(j)\neq *$ as both $z_1$ and $z_2$ are compatible with $R$. )

We call a $\wtilde{z}\in Z_R$ \emph{good} if
\[ (f|_{\wtilde{z}})(x_1)= (f|_{\wtilde{z}})(x_2)=\ldots = (f|_{\wtilde{z}})(x_l)\neq (f|_{\wtilde{z}})(x) \, .\]
Let $A$ be the set of all good $\wtilde{z}$ in $Z_R$. Notice that if $\wtilde{z}$ is good, $(( x \downarrow \wtilde{z}), L,R)$ constitutes an S-triple. We have so far shown that $z\in A$, so $A$ is non-empty. Now we prove all elements of $A$ have high degree when seen as a subset of $(n-r)$-hypercube. Indeed, notice that for any $\bar{z}\in Z_R$ and any $\bar{x}$, there are at most $s(f)$ directions $j \in [n]\setminus R$ such that $\left(f|_{\bar{z}}\right)(\bar{x}^{(j)}) \neq  \left(f|_{\bar{z}}\right)(\bar{x})$. Applying the same reasoning to all $x, x_1, x_2, \ldots, x_{l}$, we see that for any $z\in A$ there is at least $n-r-s(f)(l+1)$ neighbors of $z$ in $A$. Now applying our isoperimetric inequality (Lemma \ref{iso_theorem}) to $A$ we see that there are at least $2^{n-r-(l+1)s(f)}$ such special triples for a fixed $R\subseteq[n]$ of size $r$.

On the other hand, the number of such special triples is bounded from the above by Lemma \ref{upper_lemma}. Thus, \[ \binom{n}{r}2^{n-r-s(f)(l+1)}\leq 2^{n} \frac{ {s(f)}^{l} \, n^{r-l}} {l ! (r-l)! }. \]

As $r\ll n$ we have $\binom{n}{r}\geq \frac{n^r}{2^{r} r!}$. Simplifying we see $n^{\frac{l}{l+1}} \leq 4^r \binom{r}{l} \, s(f) 2^{s(f)}$. Choosing $r\log r= \log n$ and $l=\frac{\log r}{3}$ we get
\[ n^{1- O\left(\frac{1}{\log\log n}\right)} \leq s(f) 2^{s(f)} \, ,\]
which is our desired result.
 \end{proof}

\section{Sensitivity versus certificate complexity }\label{sec:thm1}

In this section we prove Theorem~\ref{thm:C(f)}. Actually, we prove a slightly stronger result.

\begin{Theorem}\label{thm:C(f)-2}
Let $f:\{0,1\}^n\rightarrow\{0,1\}$ be a Boolean function, then
$$
C_1(f)\leq 2^{s_0(f) - 1} s_1(f) - (s_0(f) - 1),\ C_0(f)\leq 2^{s_1(f) - 1} s_0(f) - (s_1(f) - 1).\footnote{If $s_0(f)=0$ or $s_1(f)=0$, then $f$ is constant, hence $s(f)=bs(f)=C(f)=0$.}
$$
\end{Theorem}

\proof
By symmetry we only need to prove $C_1(f)\leq 2^{s_0(f) - 1} s_1(f) - (s_0(f) - 1)$.
 Without the loss of generality, we assume $C_1(f)=C(f,0^n)$, i.e. the
1-certificate complexity is achieved on the input $0^n$. We have $f(0^n)=1$. We assume that the minimal certificate of $0^n$
consists of $x_1=0, x_2=0,\ldots,x_m=0$, where $m=C(f,0^n)=C_1(f)$.

Let $Q_0$ be the set of inputs $x$ that satisfies $x_1 =x_2=\ldots=x_{m} = 0$. Since $x_1=0, x_2=0,\ldots,x_m=0$ is a 1-certificate, we have $\forall~x\in Q_0$, $f(x)=1$.

For each $i \in [m]$, let $Q_i$ be the set of inputs $x$ with $x_1 =\ldots=x_{i-1} = x_{i+1} = \ldots= x_m = 0$ and $x_i = 1$. Let $S$ be the total sensitivity of all inputs $x \in \bigcup_{i=1}^m Q_i$. It consists of three parts: sensitivity between $Q_i$ and $Q_0$~(denoted by $S_1$), sensitivity inside $Q_i$~(denoted by $S_2$) and sensitivity between $Q_i$ and other input~(denoted by $S_3$), i.e.
\begin{equation}
S = \sum_{i=1}^{m}\sum_{x \in Q_i} s(f,x) = S_1 + S_2 + S_3.
\end{equation}

In the following we estimate $S_1,S_2$ and $S_3$ separately. We use $A_1, \ldots, A_m$ to denote the set of $0$-inputs in $Q_1, \ldots, Q_m$, respectively, i.e. $A_i=\{x\in Q_i|f(x)=0\}$~($i\in [m]$). Since $x_1=\ldots=x_m=0$ is the minimal certificate, i.e. $Q_0$ is maximal, thus $A_1, \ldots, A_m$ are all nonempty.

We also need the following lemma which follows from Lemma \ref{lem:iso} but can be
also proven without using it \cite{Sim}:
\begin{Lemma}\label{lem:Ai}
For any $i\in [m]$, $$|A_i| \geq 2^{n-m-s_0(f)+1}.$$
\end{Lemma}

The sensitivity between $Q_i$ and $Q_0$ is $|A_i|$. By summing over all possible $i$ we get:
\begin{equation}\label{eq5.3}
S_1=\sum_{i=1}^m |A_i|.
\end{equation}

{\bf Sensitivity inside $Q_1, \ldots, Q_m$:}
By Lemma~\ref{lem:iso}, for each $i\in [m]$, the number of edges between $A_i$ and $Q_i\setminus A_i$ is bounded by:
\begin{eqnarray*}
|E(A_i, Q_i\setminus A_i)| \geq  |A_i|(\log_2 |Q_i| - \log_2 |A_i|)=|A_i|(n-m-\log_2 |A_i|).
\end{eqnarray*}
Therefore,
\begin{eqnarray}
S_2&=& 2\sum_{i=1}^m |E(A_i, Q_i\setminus A_i)|\nonumber\\
&\geq &2\sum_{i=1}^m |A_i|(n-m - \log_2|A_i|).\label{eq5.4}
\end{eqnarray}

{\bf Sensitivity between $Q_i$ and other inputs (i.e.~$\{0,1\}^n\setminus \bigcup_{j=0}^{m}Q_j$):}
For each $1 \leq i < j \leq m$, let $Q_{i,j}$ be the set of inputs (not in $Q_0$) that are {\em adjacent} to both $Q_i$ and $Q_j$, i.e. $Q_{i,j}$ is the set of inputs $x$ that satisfy $x_1 =\ldots x_{i-1}=x_{i+1}=\ldots x_{j-1}=x_{j+1}=\ldots x_m = 0$ and $x_i = x_j = 1$. The sensitivity between $Q_i, Q_j$ and $Q_{i,j}$ is lower bounded by
\begin{eqnarray*}
\sum_{x \in Q_0} |f(x+e_i) - f(x+e_j)|.
\end{eqnarray*}
where $e_i$ is the unit vector with the $i$-th coordinate equal to 1 and all other coordinates equal to 0. Then,
$x+e_i$, $x+e_j$ are the neighbors of $x$ in $Q_i$ and $Q_j$, respectively. Summing over all possible pairs of $(i,j)$ we get
\begin{eqnarray}
\nonumber
S_3&\geq& \sum_{1\leq i < j\leq m} \sum_{x\in Q_0} |f(x+e_i) - f(x+e_j)|\\
& =& \sum_{x\in Q_0} \left(\sum_{i=1}^m f(x+e_i)\right)\left(m - \sum_{i=1}^m f(x+e_i)\right)\nonumber\\
&=& \sum_{x\in Q_0} s(f,x)(m-s(f,x)).\label{eq5.5}
\end{eqnarray}

\vskip10pt
If we combine inequalities~(\ref{eq5.3})-(\ref{eq5.5}), we get
\begin{eqnarray}
S &=& \sum_{i=1}^{m}\sum_{x\in Q_i}s(f,x)\nonumber\\
 &\geq& \sum_{i=1}^m |A_i| + 2\sum_{i=1}^m |A_i|(n-m - \log_2|A_i|)+ \sum_{x\in Q_0} s(f,x)(m - s(f,x)).\label{eq5.6}
\end{eqnarray}
Since $s(f,x)$ is upper bounded by $s_0(f)$ or $s_1(f)$ (depending
on whether $f(x)=0$ or $f(x)=1$), we have
\begin{eqnarray*}
\sum_{x\in Q_i}s(f,x)&\leq& |A_i|s_0(f) + (|Q_i| - |A_i|)s_1(f)\\
&=&|A_i|s_0(f) + (2^{n-m} - |A_i|)s_1(f)
\end{eqnarray*}
Thus,
\begin{eqnarray}
S=\sum_{i=1}^{m}\sum_{x\in Q_i}s(f,x) \leq \sum_{i=1}^m \Big(|A_i|s_0(f) + (2^{n-m} - |A_i|)s_1(f)\Big).\label{eq7}
\end{eqnarray}
We use $w$ to denote the total number of $0$-inputs in $Q_1, \ldots, Q_m$. Then,
\begin{eqnarray*}
w = \sum_{i=1}^m |A_i| = \sum_{x \in Q_0} s(f,x).
\end{eqnarray*}
The inequality~(\ref{eq7}) can be rewritten as
\begin{equation}
S \leq  w\cdot s_0(f) + (m \cdot 2^{n-m} -w)s_1(f).\label{eq5.7}
\end{equation}
Also, $s(f,x)\leq s_1(f)$ for each $x\in Q_0$.
Thus, the right-hand side of inequality~(\ref{eq5.6}) is
\begin{eqnarray}
\nonumber
&&\sum_{i=1}^m |A_i| + 2\sum_{i=1}^m |A_i|(n-m - \log_2|A_i|) + \sum_{x\in Q_0} s(f,x)(m - s(f,x))\\
\nonumber
 &\geq& w + 2\sum_{i=1}^m |A_i|(n-m - \log_2|A_i|) + (m - s_1(f))\sum_{x\in Q_0} s(f,x)\\
\nonumber
&= & w + 2w(n-m)-2\sum_{i=1}^m |A_i|\log_2|A_i| + (m - s_1(f))w\\
&= & w(1+2n-m - s_1(f)) - 2\sum_{i=1}^m |A_i|\log_2|A_i|.\label{eq5.8}
\end{eqnarray}
By combining inequalities~(\ref{eq5.6})-(\ref{eq5.8}) we get
\begin{eqnarray*}
w(1+2n-m - s_1(f)) - 2\sum_{i=1}^m |A_i|\log_2|A_i|
\leq w\cdot s_0(f)+(m\cdot 2^{n-m}-w)s_1(f).%=(s_1 - s_0)S + 2^V N s_0
\end{eqnarray*}
By rearranging the inequality we get
\begin{equation}
w(1 + 2n-m  - s_0(f)) \leq 2\sum_{i=1}^m |A_i| \log_2 |A_i| + m\cdot 2^{n-m} s_1(f).\label{eq8}
\end{equation}

Since the function $g(x)=x \log_2 x$ is convex and we know that
$|A_i|\leq |Q_i|=2^{n-m}$, from Lemma~\ref{lem:Ai} $|A_i|\geq 2^{n-m-s_0(f)+1}$.
Therefore,
\begin{eqnarray*}
g(|A_i|)&=&g\left(\frac{|A_i|-2^{n-m-s_0(f)+1}}{2^{n-m}-2^{n-m-s_0(f)+1}}\cdot 2^{n-m}+\frac{2^{n-m}-|A_i|}{2^{n-m}-2^{n-m-s_0(f)+1}}\cdot 2^{n-m-s_0(f)+1}\right)\\
&\leq&
\frac{|A_i|-2^{n-m-s_0(f)+1}}{2^{n-m}-2^{n-m-s_0(f)+1}}\cdot g(2^{n-m})+\frac{2^{n-m}-|A_i|}{2^{n-m}-2^{n-m-s_0(f)+1}}\cdot g(2^{n-m-s_0(f)+1})\\
&=&\frac{|A_i|-2^{n-m-s_0(f)+1}}{2^{n-m}-2^{n-m-s_0(f)+1}}\cdot 2^{n-m}(n-m)\\
&& +\frac{2^{n-m}-|A_i|}{2^{n-m}-2^{n-m-s_0(f)+1}}\cdot 2^{n-m-s_0(f)+1}(n-m-s_0(f)+1)\\
&=&\frac{|A_i|-2^{n-m-s_0(f)+1}}{2^{s_0(f)-1}-1}\cdot 2^{s_0(f)-1}(n-m) +\frac{2^{n-m}-|A_i|}{2^{s_0(f)-1}-1}(n-m-s_0(f)+1)\\
&=&\left(\frac{|A_i|-2^{n-m-s_0(f)+1}}{2^{s_0(f)-1}-1} 2^{s_0(f)-1}+\frac{2^{n-m}-|A_i|}{2^{s_0(f)-1}-1}\right)(n-m) -\frac{2^{n-m}-|A_i|}{2^{s_0(f)-1}-1}(s_0(f)-1)\\
&=&|A_i|(n-m) -\frac{2^{n-m}-|A_i|}{2^{s_0(f)-1}-1}(s_0(f)-1).
\end{eqnarray*}
Hence
\begin{eqnarray}
\nonumber \sum_{i=1}^m |A_i| \log_2|A_i|&=&\sum_{i=1}^{m}g(|A_i|)\\
\nonumber & \leq & \sum_{i=1}^m \left(|A_i|(n-m) - \frac{2^{n-m} - |A_i|}{2^{s_0(f) - 1} - 1}(s_0(f) - 1)\right)\\
&=& w(n-m+\frac{s_0(f)-1}{2^{s_0(f)-1}-1}) - m\cdot 2^{n-m}\frac{s_0(f) - 1}{2^{s_0(f) - 1} - 1}.
\label{eq9}
\end{eqnarray}
By combining inequalities~(\ref{eq8}) and (\ref{eq9}), we get
\begin{eqnarray*}
&& w(1 + 2n-m - s_0(f))\\
&& \leq 2\left(w(n-m+\frac{s_0(f)-1}{2^{s_0(f)-1}-1}) - m\cdot 2^{n-m}\frac{s_0(f) - 1}{2^{s_0(f) - 1} - 1}\right)  + m\cdot 2^{n-m} s_1(f).
\end{eqnarray*}
It implies that
\begin{eqnarray*}
w\left(1 + m - s_0(f) - \frac{2(s_0(f) - 1)}{2^{s_0(f)-1} - 1}\right) \leq m\cdot 2^{n-m} \left(s_1(f) - \frac{2(s_0(f) - 1)}{2^{s_0(f)-1} - 1}\right),
\end{eqnarray*}
Substituting $w =\sum_{i=1}^{m} |A_i|\geq m \cdot 2^{n-m-s_0(f)+1}$, we get
\begin{eqnarray*}
m\cdot 2^{n-m-s_0(f)+1} \left(1 + m - s_0(f) - \frac{2(s_0(f) - 1)}{2^{s_0(f)-1} - 1}\right) \leq m\cdot 2^{n-m} \left(s_1(f) - \frac{2(s_0(f) - 1)}{2^{s_0(f)-1} - 1}\right),
\end{eqnarray*}
i.e.
\begin{eqnarray*}
1 + m - s_0(f) - \frac{2(s_0(f) - 1)}{2^{s_0(f)-1} - 1} \leq 2^{s_0(f)-1} \left(s_1(f) - \frac{2(s_0(f) - 1)}{2^{s_0(f)-1} - 1}\right),
\end{eqnarray*}
which implies
$$m \leq 2^{s_0(f) - 1} s_1(f) - s_0(f) + 1.$$
\qed

\subsection{Proof of Corollary~\ref{thm:bs(f)}}\label{sec:cor}

To prove Corollary~\ref{thm:bs(f)}, we need the following Lemma by Kenyon and Kutin.\footnote{In their original paper there is no ``$-\frac{1}{2}$" term, but a careful analysis will provide it.}
\begin{Lemma}\cite{KK}\label{lemm:KK}
\label{cor:cbs2}
$bs_0(f) \leq 2 ( C_1(f) - \frac{1}{2}) s_0(f),\ bs_1(f) \leq 2 ( C_0(f) - \frac{1}{2}) s_1(f)$.
\end{Lemma}

\proof (of Corollary~\ref{thm:bs(f)}) From Theorem~\ref{thm:C(f)}, we have $bs_0(f)\leq C_0(f)\leq 2^{s_1(f)-1}s_0(f).$
From Corollary~\ref{cor:cbs2} we have $bs_0(f)\leq 2(C_1(f) - \frac{1}{2}) s_0(f)$, together with Theorem~\ref{thm:C(f)} we get
$bs_0(f)\leq 2 (2^{s_0(f)-1}s_1(f)-\frac{1}{2}) s_0(f).$
Therefore,
$bs_0(f)\leq \min\{2^{s_1(f)}s_0(f),\ 2^{s_0(f)}s_1(f)s_0(f)\}.$
Similarly we can show that
$bs_1(f)\leq \min\{2^{s_1(f)}s_0(f)s_1(f),\ 2^{s_0(f)}s_1(f)\}.$
\qed

\section{Concluding remarks}
In this work we presented some results toward the sensitivity conjecture providing the first improvement since the work of Kenyon and Kutin \cite{KK}. It is certainly desirable to understand the limits of the techniques introduced here. Another interesting problem is to unify our approaches in Sections \ref{sens_degree_section} and \ref{sec:thm1}. Although, the structure of the two proofs appear quite different, the fact that they both crucially rely on Harper's isoperimetric inequality might be hinting at a more explicit relationship between the two.
 
In this work we have been mostly concerned with the original formulation of the Sensitivity Conjecture in terms of complexity measures of Boolean functions. Before ending this paper however, we would like to take the opportunity to recount a purely combinatorial formulation of the problem which may be more accessible to the wider mathematics community.

It turns out that the sensitivity conjecture is equivalent to the validity of the following Ramsey-type statement: There exists a constant $\delta>0$ such that for any unbalanced two-coloring of vertices of hypercube $\{0,1\}^n$ contains a vertex $x\in\{0,1\}^n$ such that $x$ has  $\geq n^\delta$ neighbors in the same color class as $x$. Implicit in the above statement is the following observation: It is rather easy to construct a balanced two-coloring of the vertices of Hamming cube, i.e.~ each of size $2^{n-1}$, such that each point $x\in \{0,1\}^n$ would have only the elements of the other color class as its neighbors; this can be seen by putting the points with odd Hamming weight in one class, and the ones with even Hamming weight in the other. However, after trying to find subsets of slightly larger than half with small maximum degree, one soon realizes that such sets are indeed hard to construct. \footnote{For instance, in the above example if one transfers even a single odd weight point to the set of even weight points of discrete cube, the resulting induced subgraph will have max-degree $n$.}

The above discussion provides further evidence for the well-known intuition that averaging type arguments (including most purely Fourier analytic ones) are hopeless in addressing the conjecture without further input. However, at the moment it remains unclear what type of extra input one may need, beside the well-known Fourier analytic ones, to tackle the conjecture.

\end{document}